\pdfoutput=1
\documentclass[journal]{IEEEtran}  

\usepackage{cite}
\usepackage{amsmath}
\usepackage{fancyhdr}
\usepackage{graphicx}
\usepackage{amssymb}
\usepackage{times}
\usepackage{graphics} 
\usepackage{epsf} 
\usepackage{times} 
\usepackage{amsmath} 
\usepackage{amssymb}  
\usepackage{relsize}
\usepackage{algorithmicx}
\usepackage[ruled]{algorithm}
\usepackage{algpseudocode}
\usepackage{epstopdf}
\usepackage{subfigure,multirow}
\usepackage{float}
\usepackage{bm}
\usepackage{mathrsfs}

\newtheorem{rmk}{Remark}
\newtheorem{pro}{Proposition}
\newtheorem{dfn}{Definition}
\newtheorem{asm}{Assumption}
\newtheorem{thm}{Theorem}
\newtheorem{prp}{Property}
\newtheorem{proof}{Proof}
\newtheorem{proofpart}{Part}

\makeatletter
\@addtoreset{proofpart}{thm}
\makeatother

\begin{document}
	
	\title{	Data-driven feedback stabilization of nonlinear systems: Koopman-based model predictive control
	}
	
	\author{ Abhinav Narasingam and Joseph Sang-Il Kwon
		\thanks{Abhinav Narasingam, and Joseph Sang-Il Kwon (email: kwonx075@tamu.edu) are with the Artie McFerrin Department of Chemical Engineering, Texas A\&M University, College Station, TX 77845 USA
	}
 }
	\maketitle

\begin{abstract}
In this work, a predictive control framework is presented for feedback stabilization of nonlinear systems. To achieve this, we integrate Koopman operator theory with Lyapunov-based model predictive control (LMPC). The main idea is to transform nonlinear dynamics from state-space to function space using Koopman eigenfunctions - for control affine systems this results in a bilinear model in the (lifted) function space. Then, a predictive controller is formulated in Koopman eigenfunction coordinates which uses an auxiliary Control Lyapunov Function (CLF) based bounded controller as a constraint to ensure stability of the Koopman system in the function space. Provided there exists a continuously differentiable inverse mapping between the original state-space and (lifted) function space, we show that the designed controller is capable of translating the feedback stabilizability of the Koopman bilinear system to the original nonlinear system. Remarkably, the feedback control design proposed in this work remains completely data-driven and does not require any explicit knowledge of the original system. Furthermore, due to the bilinear structure of the Koopman model, seeking a CLF is no longer a bottleneck for LMPC. Benchmark numerical examples demonstrate the utility of the proposed feedback control design.
\end{abstract}


 \ifCLASSOPTIONpeerreview
 \begin{center} \bfseries EDICS Category: 3-BBND \end{center}
 \fi
%
\IEEEpeerreviewmaketitle
\section{Introduction}	
\IEEEPARstart{N}{onlinear} systems abound in nature, and providing a universal feedback design for stabilizing general nonlinear dynamics thus stands to have a significant impact on a broad range of applications. Yet, it remains a daunting challenge owing to the complexity of these nonlinear models. Often, powerful tools from differential geometry are required to fully resolve the complexities, making them computationally intractable. State-space models are one of the widely used ways to represent these dynamics, 
and several existing approaches that use the state-space description for nonlinear stabilizing control include optimization-based Sum of Squares (SoS) \cite{Parillo2000}, geometric-based feedback linearization \cite{Astolfi2015}, sliding mode control \cite{Utkin1993}, etc. An alternative to the state-space description of dynamical systems is the operator-theoretic description where we are interested in the evolution of observables (functions of states) and not the states themselves. Introduced by Koopman in \cite{Koopman1931}, the Koopman operator is one example which, when acted upon an observable, governs its evolution along the original system trajectory. Hence, the operator-theoretic description provides global insight into the system dynamics. This makes the Koopman operator a natural choice for data-driven analysis of dynamical systems and is appropriate for controller design.

Another feature of the Koopman operator approach that makes it extremely appealing for controller design is that it is a linear operator, although infinite dimensional, even when associated with nonlinear dynamics. Thus it extends the spectral analysis concepts of linear systems to the dynamics of observables in nonlinear systems. Specifically, the eigenvalues, eigenfunctions and invariant subspaces of the Koopman operator encode global information and provide valuable insights that allows future state prediction and scalable reconstruction of the underlying dynamics \cite{Budisic2012}. It has been shown in the literature that Koopman eigenfunctions are strongly connected to the geometric properties of the system \cite{Mauroy2012,Mauroy2013}, and are related to global linearization of the system \cite{Lan2013}. Recently, the connection between (existence of) specific Koopman eigenfunctions and the global
stability analysis of has been explored \cite{Mauroy2016}. Koopman eigenfunctions and the corresponding eigenvalues also facilitate the estimation of limit cycles as well as their basins of attraction \cite{Mezic2017}. Furthermore, the full state of the system can be projected onto the eigenfunctions of the Koopman operator using a linear combination of (Koopman) modes characterized by a fixed frequency and rate of decay. Therefore, dominant patterns of the underlying nonlinear system can be captured using these modes as useful coherent structures \cite{Mezic2005}. 

The practical implementation of operator theoretic control has been driven by the emergence of advanced numerical algorithms that approximate the Koopman spectrum from time-series data such as Dynamic Mode Decomposition \cite{Tu2014}, Extended Dynamic Mode Decomposition (EDMD) \cite{Williams2015}, Laplace analysis \cite{Mohr2014}, and machine learning \cite{Lusch2018,Li2017,Otto2019}. These developments have made Koopman operator theory an increasingly attractive approach for the analysis and control of nonlinear dynamical systems \cite{Korda2018a,Korda2018c,Arbabi2018,Hanke2018}. Although successfully implemented on a broad range of applications, the true potential of Koopman approach can only be realized by certifying that the controllers will guarantee closed-loop stability and robustness. Unlike systems characterized by unforced dynamics, providing stability analysis for forced (input dependent) systems has proven to be difficult because the predictive capability of the Koopman operator can be significantly impacted unless the role of actuation (i.e., the manipulated inputs) is appropriately accounted. To deal with this, \cite{Proctor2018} redefined the Koopman operator as a function of both states and the inputs. In \cite{Williams2016}, a modification of EDMD was presented that compensates for the effect of inputs. In \cite{Surana2016b}, a bilinear representation was provided in the Koopman space that is tight and theoretically justified. Using this representation, the authors in \cite{Huang2018} proposed a stabilizing feedback controller which relies on control Lyapunov function (CLF) and thus achieves stabilization of the bilinear system. 

However, the method in \cite{Huang2018} neither solved an optimal control problem nor accounted for explicit state and input constraints. Moreover, it did not comment on the stability analysis of the original nonlinear system under the implementation of the designed controller in the Koopman space. To address this, CLFs were employed in \cite{Narasingam2019} where a feedback controller was designed for the Koopman space (i.e., lifted domain) using Lyapunov constraints within a model predictive control (MPC) formulation. Such a design allowed for an explicit characterization of stability properties of the original nonlinear system. In addition, the linear structure of the Koopman models was exploited to transform the original nonlinear MPC problem to a convex quadratic MPC problem that is computationally attractive. However, the limitation of the method presented in \cite{Narasingam2019} is that the CLF was derived for the original system which requires an explicit mathematical expression of the original nonlinear dynamics; it is particularly challenging when we have limited a priori knowledge of the original nonlinear system. Additionally, even though we have a good understanding of the nonlinear system, it is in practice computationally demanding to determine its corresponding CLFs.  

To address these issues, this work seeks to propose a stabilizing feedback controller based on the Koopman bilinear representation of the original nonlinear system. To do so, first, Koopman system identification is applied to derive a bilinear representation of the dynamics. Then, a CLF is determined for the bilinear system in the Koopman eigenfunction space which is employed in the Lyapunov-based MPC (LMPC) formulation. Then, a stability criterion is presented that guarantees stability of the original closed-loop system in the $\epsilon-\delta$ sense based on stability of the Koopman bilinear system. Unlike \cite{Narasingam2019}, the feedback control design proposed in this work is completely data-driven and does not require any a priori knowledge of the original system. Moreover, deriving CLFs for the Koopman bilinear system is much more computationally affordable than the original nonlinear system. In fact, the search for CLFs can be focused on a class of quadratic functions which are known to effectively characterize the stability region of simpler systems like the (Koopman) bilinear systems.     
       
\emph{Organization}: Section \ref{Preliminaries} contains definitions of mathematical concepts of interest and describes the Koopman system identification method. In Section \ref{KLMPC}, we present our stabilizing controller formulation, utilizing the identified model, based on Lyapunov-based MPC and prove stability of the closed-loop system under the implementation of the proposed controller. Section \ref{Results} illustrates the application to numerical examples and the performance of the controller is studied. Section \ref{Conclusions} provides concluding remarks and discusses future work

\section{Preliminaries}\label{Preliminaries}
In this section, we provide background on the Koopman operator and its relation to forced dynamical systems. Subsequently, we present a system identification method over Koopman observables, which yields a practical training procedure for embedding nonlinear systems to a bilinear model from data.

\subsection{Koopman Operator}
Let $\mathbf{x} \in \mathcal{X} \subseteq \mathbb{R}^n$ be the vector of state variables  of a continuous-time nonlinear dynamical system whose evolution is governed by the function
\begin{equation}\label{eq1}
\dot{\mathbf{x}} = \mathbf{F}(\mathbf{x})
\end{equation}
where $\mathbf{F} : \mathcal{X} \to \mathcal{X}$ is the nonlinear operator that maps the system states forward in time. It is assumed that the vector field $\mathbf{F}$ is continuously differentiable. The solution to \eqref{eq1} is given by the flow field $\mathbf{\Phi}^t(\mathbf{x})$. Typically, an analytic form for $\mathbf{\Phi}^t(\mathbf{x})$ is impossible to determine and we resort to numerical solutions for \eqref{eq1}, which can become computationally intractable.  

Now, let $\mathcal{G}$ be a Hilbert space of complex-valued functions on $\mathcal{X}$. The elements of $\mathcal{G}$ are often called \emph{observables}  as they may correspond to measurements taken during an experiment or the output of
a simulation. In his seminal work, Koopman realized an alternative description of (\ref{eq1}) in terms of the evolution of these observables denoted as $g(\mathbf{x})$ with $g:  \mathcal{X} \to \mathbb{C}$. Specifically, Koopman theory asserts that the nonlinear system in \eqref{eq1} can be mapped to a \emph{linear} system using an infinite dimensional linear operator $\mathcal{K}^t$ that advances these observables forward in time. 
\begin{dfn}[Koopman operator]
	For a given space $\mathcal{G}$ of observables, the Koopman (semi)group of operators $ \mathcal{K}^t : \mathcal{G} \to \mathcal{G}$ associated with system (\ref{eq1}) is defined by
	\begin{equation}\label{eq2}
	 [\mathcal{K}^tg](\mathbf{x}) = g \circ \mathbf{\Phi}^t(\mathbf{x})
	\end{equation}
\end{dfn}
By definition, the Koopman operator is linear even though the underlying dynamical system is nonlinear, i.e., it satisfies
\begin{equation}\label{eq3}
[\mathcal{K}^t(\alpha g_1 + \beta g_2)](\mathbf{x}) = \alpha[\mathcal{K}^tg_1](\mathbf{x}) + \beta[\mathcal{K}^tg_2](\mathbf{x})
\end{equation}
The linearity of the Koopman operator allows it to be characterized by its eigenvalues and eigenfunctions. An eigenfunction $\psi \in \mathcal{G} : \mathcal{X} \to \mathbb{C}$ of the Koopman operator is defined to satisfy
\begin{equation}\label{eq4}
\begin{aligned}
& [\mathcal{K}^t\psi](\mathbf{x}) = e^{\lambda t} \psi(\mathbf{x}) \\
& \frac{d}{dt}\psi(\mathbf{x}) = \lambda \psi(\mathbf{x})
\end{aligned}
\end{equation} 
where $\lambda \in \mathbb{C}$ is the associated eigenvalue. These eigenfunctions can be used to predict the time evolution of an observable, in relation with the state dynamics, as long as the given observable lies within the span of these eigenfunctions. Applying chain rule to \eqref{eq4}, 

\begin{equation}\label{eq5}
\begin{aligned}
& \frac{d}{dt}\psi(\mathbf{x}) =  \nabla \psi(\mathbf{x}) \cdot \mathbf{F}(\mathbf{x}) \triangleq L_\mathbf{F} \psi(\mathbf{x}) = \lambda \psi(\mathbf{x})
\end{aligned}
\end{equation}
where the Lie derivative with respect to the vector field $\mathbf{F}$, denoted as $L_\mathbf{F} = \mathbf{F} \cdot \nabla$, is the infinitesimal generator of the Koopman operator $\mathcal{K}^t$, i.e., $\lim\limits_{t \to 0}(\mathcal{K}^t - I)/t$. Hence, the time varying observable $\tilde{g}(t,\mathbf{x}) = \mathcal{K}^tg(\mathbf{x})$ can be obtained as a solution to the partial differential equation,

\begin{equation}\label{eq6}
\begin{aligned}
& \frac{\partial}{\partial t}\tilde{g} = \mathbf{F} \cdot \nabla \tilde{g} \triangleq L_\mathbf{F}\tilde{g} \\
& \tilde{g}(0,\mathbf{x}) = g(\mathbf{x})
\end{aligned}
\end{equation} 

Any finite subset of the Koopman eigenfunctions naturally forms an invariant subspace and discovering these eigenfunctions enables globally linear representations of strongly nonlinear systems.

\subsection{Modeling forced dynamics}\label{forced}
The Koopman operator theory has been conceptually developed for uncontrolled systems. To adopt it for the purposes of control, consider a control affine system as follows:
\begin{equation}\label{eq7}
\dot{\mathbf{x}} = \mathbf{F}(\mathbf{x}) + \sum_{i=1}^{m}\mathbf{G}_i(\mathbf{x})u_i
\end{equation}
where $\mathbf{x} \in \mathcal{X} \subseteq \mathbb{R}^n$, $u_i \in \mathcal{U}$ for $i=1,\cdots,m$, and $\mathbf{G}_i : \mathcal{X} \to \mathcal{X}$ denotes the control vector fields that dictate the effect of input on the system. It is assumed that the vector fields are locally Lipschitz continuous. This is a reasonable assumption which holds true for many physical systems. The evolution of the observable functions for the controlled system of \eqref{eq7} is given, by applying chain rule similar to \eqref{eq6}, as  

\begin{equation}\label{eq8}
\begin{aligned}
& \frac{\partial}{\partial t}\tilde{g} = L_\mathbf{F}\tilde{g} + \sum_{i=1}^{m}u_iL_{\mathbf{G}_i}\tilde{g} \\
& \tilde{g}(0,\mathbf{x}) = g(\mathbf{x})
\end{aligned}
\end{equation} 
where $L_\mathbf{F}$ and $L_{\mathbf{G}_i}$ denote the Lie derivatives with respect to the vector fields $\mathbf{F}$ and $\mathbf{G}_i$ for $i = 1,\cdots,m$, respectively. The system \eqref{eq8} is analogous to a bilinear system except for the fact that the operators $L_\mathbf{F}$ and $L_{\mathbf{G}_i}$ are infinite dimensional, operating on the function space $\mathcal{G}$. However, if there exist a finite number of observable functions $\bar{g}_1,\cdots, \bar{g}_N$ that span a subspace $\bar{\mathcal{G}} \subset \mathcal{G}$ such that $\mathcal{K}\bar{g} \in \bar{\mathcal{G}}$ for any $\bar{g} \in \bar{\mathcal{G}}$, then $\bar{\mathcal{G}}$ is said to be an invariant subspace and the Koopman operator becomes a finite-dimensional matrix, $\mathbf{K}$. For practical implementation, the Koopman eigenfunctions can be used for $\bar{g}$ such that the finite-dimensional approximation can be determined by projecting on the subspace spanned by these eigenfunctions \cite{Goswami2017}. The choice of using eigenfunctions as basis is intuitive because an action of the infinitesimal generator of the Koopman operator on these eigenfunctions is dictated simply by a scalar, i.e., the corresponding eigenvalue (see \eqref{eq5}).

\subsection{Koopman bilinear system identification }\label{EDMD}
To obtain a bilinear form of system \eqref{eq7} in the Koopman eigenfunction coordinates, we use the Koopman canonical transform (KCT) \cite{Surana2016b}. Such a transformation is given by

\begin{equation}\label{eq9}
\begin{aligned}
 \mathbf{z} = \boldsymbol{\Psi}(\mathbf{x}) =& ~[\psi_1(\mathbf{x}),\cdots,\psi_N(\mathbf{x})]^T, \text{where}\\
 \psi_j(\mathbf{x}) =& ~\tilde{\psi}_j(\mathbf{x}),  \text{if} ~\tilde{\psi}_j : \mathcal{X} \to \mathbb{R} \\
 [\psi_j(\mathbf{x}),~\psi_{j+1}(\mathbf{x})]^T =& ~[2Re(\tilde{\psi}_j(\mathbf{x})), ~-2Im(\tilde{\psi}_j(\mathbf{x}))]^T, \\
&  ~\text{if} ~\tilde{\psi}_j, ~\tilde{\psi}_{j+1} : \mathcal{X} \to \mathbb{C} \\
& ~\text{and assuming} ~ \tilde{\psi}_{j+1} = { \tilde{\psi}_j^\star }
\end{aligned}
\end{equation}
where $^\star $ denotes the complex conjugate. 
Applying the above transformation to \eqref{eq7} yields
\begin{equation}\label{eq10}
\dot{\mathbf{z}} = \Lambda \mathbf{z} + \sum_{i=1}^{m}u_iL_{\mathbf{G}_i}\mathbf{\Psi}
\end{equation}
where $\Lambda$ is a block-diagonal matrix constructed using the Koopman eigenvalues $\lambda_j, j = 1,\cdots, N$, which are corresponding to the Koopman eigenfunctions shown in \eqref{eq9}, i.e., 
\begin{equation}\label{eq12}
\begin{aligned}
\Lambda_{j,j} &= \lambda_j, \hspace{0.5cm} \text{if} ~ ~\tilde{\psi}_j : \mathcal{X} \to \mathbb{R} \\
\begin{bmatrix}
\Lambda_{j,j} & \Lambda_{j,j+1} \\
\Lambda_{j+1,j} & \Lambda_{j+1,j+1} \\
\end{bmatrix} &= |\lambda_j|\begin{bmatrix}
cos(\angle\lambda_j) & sin(\angle\lambda_j)\\
-sin(\angle\lambda_j) & cos(\angle\lambda_j)
\end{bmatrix},\\
& \hspace{1.5cm} \text{if} ~ \tilde{\psi}_{j} ,\tilde{\psi}_{j+1} : \mathcal{X} \to \mathbb{C} \\
& \hspace{1.5cm}\text{and assuming} ~ \tilde{\psi}_{j+1} = { \tilde{\psi}_j^\star }
\end{aligned}
\end{equation}
\begin{asm}\label{asm1}
$\exists ~\psi_j, j = 1,\cdots, N$ such that \\
\begin{equation*}
L_{\mathbf{G}_i}\mathbf{\Psi} = \sum_{j = 1}^{N} b^{\mathbf{G}_i}_j\psi_j(\mathbf{x}) = B_i\boldsymbol{\Psi}
\end{equation*}
\end{asm}
where $b^{\mathbf{G}_i}_j \in \mathbb{R}^n$ and $\psi_j(\mathbf{x})$ are defined in \eqref{eq9}. In other words, it is assumed that $L_{\mathbf{G}_i}\boldsymbol{\Psi}$ lies in the span of the eigenfunctions $\psi_j, ~j = 1,\cdots,N$ so that it can be represented using a constant matrix, $B_i \in \mathbb{R}^{N \times N}$.

Based on this assumption, the system \eqref{eq10} becomes the following bilinear control system in the Koopman space,
\begin{equation}\label{eq11}
\dot{\mathbf{z}} = \Lambda\mathbf{z} + \sum_{i=1}^{m}u_iB_i\mathbf{z}
\end{equation}

The objective of the system identification method is to determine the continuous bilinear system of \eqref{eq11} using time-series data generated by the controlled dynamical system of \eqref{eq7}. This is done in two parts. First, we calculate the system matrix $\Lambda$ using the eigenfunctions of the Koopman operator for the uncontrolled part of \eqref{eq7}. Although there are several  methods available in the literature that can achieve this, the EDMD algorithm is utilized in this work. The algorithm is detailed below. \\
\emph{Calculating $\Lambda$}:
\begin{enumerate}
	\item The time-series data of $N_t$ snapshot pairs satisfying the dynamical system of (\ref{eq1}) are generated and organized in the following matrices:
	\begin{equation}\label{eq13}
	\begin{aligned}
	&	\mathbf{X} = [\mathbf{x}_1,\mathbf{x}_2,\cdots,\mathbf{x}_{N_t}],  \quad \mathbf{Y} = [\mathbf{y}_1,\mathbf{y}_2,\cdots,\mathbf{y}_{N_t}], \\
	\end{aligned}
	\end{equation}
	where $\mathbf{x}_k \in \mathcal{X}$, $\mathbf{y}_k = \mathbf{F}(\mathbf{x}_k)\Delta t + \mathbf{x}_k \in \mathcal{X}$ and $\Delta t$ is the discretization time. Note $\mathbf{y}_k$ is used here instead of $\mathbf{x}_{k+1}$ because the data need not necessarily be temporally ordered as long as the corresponding pairs $(\mathbf{x}_k,\mathbf{y}_k)$ are obtained as shown above.
	\item  A library of nonlinear observable functions $\mathcal{D} = \{\phi_1,\phi_2,\ldots,\phi_N\}$ is selected to define the vector-valued function $\boldsymbol{\phi}: \mathcal{X} \to \mathbb{R}^N$ 
	\begin{equation}\label{eq14}
	\boldsymbol{\phi}(\mathbf{x}) = [\phi_1(\mathbf{x}),\phi_2(\mathbf{x}),\cdots,\phi_{N}(\mathbf{x})]^T
	\end{equation}
	where $\boldsymbol{\phi}$ is used to lift the system from a state space to a function space of observables. 
 \item A least-squares problem is solved over all the data samples to obtain $\mathbf{K} \in \mathbb{R}^{N \times N}$ which is the transpose of the finite dimensional approximation to the Koopman operator, $\mathcal{K}^t$:
\begin{equation}\label{eq15}
\underset{\mathbf{K}}{\text{min}} \quad \sum_{i = 1}^{N_t}\| \boldsymbol{\phi}(\mathbf{y}_i)-\mathbf{K}\boldsymbol{\phi}(\mathbf{x}_i)\|_2^2
\end{equation}
The  value of $\mathbf{K}$ that minimizes (\ref{eq15}) can be determined analytically as:
\begin{equation}\label{eq16}
\mathbf{K} = \boldsymbol{\phi}_{XY}\boldsymbol{\phi}_{XX}^{\dagger}
\end{equation}
where $^{\dagger}$ denotes the pseudo inverse, and the data matrices are given by 
\begin{equation}\label{eq17}
\begin{aligned}
& \boldsymbol{\phi}_{XX} = \boldsymbol{\phi}_X \boldsymbol{\phi}_X^T , \quad \boldsymbol{\phi}_{XY} = \boldsymbol{\phi}_Y \boldsymbol{\phi}_X^T \\
\end{aligned}
\end{equation}
where
\begin{equation*}
\begin{aligned}
& \boldsymbol{\phi}_X = [\boldsymbol{\phi}(\mathbf{x}_1),\cdots,\boldsymbol{\phi}(\mathbf{x}_{N_t})], \\
&  \boldsymbol{\phi}_Y = [\boldsymbol{\phi}(\mathbf{y}_1),\cdots,\boldsymbol{\phi}(\mathbf{y}_{N_t})]
\end{aligned}
\end{equation*}
It has been previously shown that the matrix $\mathbf{K}$ asymptotically approaches the Koopman operator as we increase $N_t$ \cite{Korda2018b}, and hence approximates the evolution of observables.
\item An eigendecomposition of $\mathbf{K}$ is performed to determine the eigenvalues $\tilde{\lambda}_j$ and eigenvectors $e_j$ for $j = 1,\cdots, N$.
\item The eigenvalues are converted to continuous time as $\lambda_j = \text{log}(\tilde{\lambda}_j)/\Delta t$, and the eigenfunctions, $\psi_j$, are computed, using  $\tilde{\psi}_j = \boldsymbol{\phi}^Te_j$, according to the procedure described in \eqref{eq9}.
\item The system matrix $\Lambda$ is constructed using the block-diagonalization described in \eqref{eq12}.
\end{enumerate}
\emph{Calculating control matrix $B_i$}:\\
In the next step, the control matrix $B_i$ is calculated using the eigenfunctions. Specifically using \emph{Assumption \ref{asm1}} and the fact that $\boldsymbol{\Psi}(\mathbf{x}) = E^T\boldsymbol{\phi}(\mathbf{x})$ where $E = [e_1,\cdots,e_N]$ is the matrix containing the eigenvectors, we have 
\begin{equation}\label{eq18}
\begin{aligned}
B_i\boldsymbol{\Psi}(\mathbf{x})  &= L_{\mathbf{G}_i}\boldsymbol{\Psi}(\mathbf{x}) \\
&= L_{\mathbf{G}_i}(E^T\boldsymbol{\phi}(\mathbf{x})) = E^T\frac{\partial\boldsymbol{\phi}}{\partial \mathbf{x}}\mathbf{G}_i(\mathbf{x})
\end{aligned}
\end{equation}
The control matrix $B_i$ can be obtained by equating the coefficients of right and left hand side functions of the above equation. Once the system matrices $\Lambda$ and $B_i$ are determined, a bilinear system of \eqref{eq11} can be constructed using the Koopman eigenfunctions and can be used for the task of designing feedback controllers.

\section{Koopman Lyapunov-based MPC} \label{KLMPC}
In this section, we detail how Koopman operator theory can be integrated with Lyapunov-based predictive control scheme to stabilize the system of \eqref{eq7}
\subsection{Lyapunov-based predictive control}

For simplicity, let us consider the control affine system of \eqref{eq7} with $i=1$, i.e., a single input. All the results can be generalized to the case of multiple inputs. Without loss of generality, we assume $\mathbf{F}({0}) = {0}$ and that the origin is an unstable equilibrium point of the uncontrolled system. Then, the closed-loop stabilization problem associated with \eqref{eq7} seeks a state-dependent control law of the form $u=h(\mathbf{x}), h : \mathbb{R}^n \to \mathbb{R}$ which renders the origin stable within some domain $\mathscr{D} \subset \mathbb{R}^n$ for the closed-loop form of \eqref{eq7}. 

One of the widely used approaches to design state feedback controllers is via the use of  CLFs as they facilitate explicit consideration of the stability prior to the controller design. CLF is a continuously differentiable positive definite function $V: \mathscr{D} \to \mathbb{R}_+$ such that for all $\mathbf{x} \in \mathscr{D}/\{0\}$, $\dot{V} := L_\mathbf{F}V + uL_\mathbf{G}V < 0$. Once  this  CLF  is  constructed,  design  of  a  feedback  law  can  be straightforward  \cite{Sontag1989}.

LMPC is a powerful tool that uses CLFs for the design of an optimal stabilizing feedback controller for nonlinear dynamical systems \cite{Mhaskar2005}, particularly those characterized by a set of constraints. Essentially, LMPC is a control strategy that possesses all the advantages of a standard MPC and is designed based on an explicit, stable (albeit not optimal) control law ${h}(\cdot)$. By explicitly adding a Lyapunov constraint to a standard MPC formulation, the controller is able to stabilize the closed-loop system.  Additionally, LMPC explicitly characterizes a set of initial conditions starting from where the closed–loop stability is guaranteed. Hence, it ensures stability irrespective of the prediction horizon, i.e., the computational time can be made smaller by decreasing the prediction horizon (reducing the size of the optimization problem). However, the main bottleneck to the success of this method lies in the construction  of CLFs for a general nonlinear system. To avoid this, in the proposed method, the system of \eqref{eq7} is first transformed into a bilinear control system of \eqref{eq11}, using the procedure described above, for which determining a CLF is much easier. Particularly, the search for a CLF of a bilinear system can now be limited to the class of quadratic functions and an optimization problem can be solved to determine the required CLF \cite{Huang2018}. Then, one can apply LMPC in the Koopman eigenspace to determine a stabilizing input for the bilinear system of \eqref{eq11}. 

\subsection{Bounded explicit control $h(\mathbf{z})$}
Let us consider the Koopman bilinear system of \eqref{eq11} with $i=1$, i.e., a single input obtained using the system identification method described in Section \ref{EDMD}. This system is assumed to be stabilizable, which implies the existence of a feedback control law ${u}(t) = {h}(\mathbf{z}(t))$ that satisfies input constraints for all $\mathbf{z}$ inside a given stability region and renders the origin of the closed-loop system asymptotically stable. This is equivalent to assuming that there exists a CLF for the system of \eqref{eq11}. Due to the bilinear structure of the system, the CLF can be limited to a class of quadratic functions, i.e., $V(\mathbf{z}) = \mathbf{z}^TP\mathbf{z}$. The necessary and sufficient conditions for the symmetric positive definite matrix $P$ such that the system of \eqref{eq11} is stabilizable are provided in \cite{Huang2018}. The theorem is stated below. 

\begin{pro}[see \cite{Huang2018}, Theorem 2]\label{pro1}
	The bilinear system of \eqref{eq11} is stabilizable if and only if there exists an $N\times N$ symmetric positive definite matrix $P$ such that for all $\mathbf{z} \neq 0 \in \mathbb{R}^N$ with $\mathbf{z}^T(P\Lambda + \Lambda^TP)\mathbf{z} \geq 0$, we have $\mathbf{z}^T(PB + B^TP)\mathbf{z} \neq 0$.
\end{pro}  

In other words, for $\dot{V}(\mathbf{z}) = \mathbf{z}^T(P\Lambda + \Lambda^TP)\mathbf{z} + u(\mathbf{z}^T(PB + B^TP)\mathbf{z})$ to be negative, given that the first term on the right hand side is positive, then the second term cannot be zero so that the control action $u$ can render $\dot{V} < 0$. Once the conditions of \emph{Proposition \ref{pro1}} are satisfied, one way to determine the explicit control law $h(\mathbf{z})$, required to stabilize the bilinear system, is provided by the following formula by Sontag \cite{Lin1991}:

\begin{equation}\label{eq19}
\begin{aligned}
& b(\mathbf{z}) = 
\begin{cases}
-\frac{L_\Lambda V + \sqrt{L_\Lambda V^2 + L_BV^4}}{L_BV}, & \text{if} ~~ L_BV \neq 0 \\[5pt]
0, & \text{if} ~~L_BV = 0
\end{cases} \\[5pt]
& h(\mathbf{z}) = \begin{cases}
u_{min} , & \text{if} ~~ b(\mathbf{z}) < u_{min} \\[5pt]
b(\mathbf{z}), & \text{if} ~~ u_{min} \leq b(\mathbf{z}) \leq u_{max} \\[5pt]
u_{max}, & \text{if}~~ b(\mathbf{z}) > u_{max}
\end{cases}
\end{aligned}
\end{equation}
where $L_\Lambda V = \mathbf{z}^T(P\Lambda + \Lambda^TP)\mathbf{z}$, $L_BV = \mathbf{z}^T(PB + B^TP)\mathbf{z}$, and $h(\mathbf{z})$ represents the saturated control law that accounts for the input constraints $u_{min} \leq u(t) \leq u_{max} \in \mathcal{U}$. For the above controller, one can show, using a standard Lyapunov argument, that if the closed-loop state evolves within a level set of $V$, the time-derivative of the CLF is negative definite ensuring asymptotic stability. Let the largest level set of $V$ be given by

\begin{equation}\label{eq20}
\Omega_r = \{\mathbf{z} \in \mathbb{R}^N : V(\mathbf{z}) \leq r\}
\end{equation} 
where $r$ is the largest number for which $\Omega_r \subseteq \Omega$. $\Omega$ is the complete stability region, starting from which the origin of the bilinear system under \eqref{eq19} is guaranteed to be stable. In practice, the entire region of attraction, $\Omega$, is very difficult to estimate even for simple systems.
 
\subsection{Koopman Lypaunov-based predictive control}
Now that we have the explicit control law, the idea is to stabilize the bilinear system using the Lyapunov-based predictive control scheme as below:
\begin{subequations}\label{eq21}
	\begin{align}
	& \underset{{u}\in \mathcal{S}(\Delta)}{\text{min}} 
	&& \int_{t_k}^{t_{k}+N_p\Delta} [\mathbf{z}^T(\tau)W\mathbf{z}(\tau) + {u}^T(\tau)Ru(\tau)]d\tau, \label{eq21a}\\[7pt]
	&\text{s.t}
	&& \dot{{\mathbf{z}}}(t) = \Lambda\mathbf{z}(t) + u(t){B}\mathbf{z}(t)\label{eq21b}\\
	&&&{\mathbf{z}}(t_k) = \boldsymbol{\Psi}(\mathbf{x}(t_k)) \label{eq21c}\\
	&&& u_{min} \leq u(t) \leq u_{max}, ~\forall t \in [t_k,t_{k}+N_p\Delta) \label{eq21d}\\
	&&& V(\mathbf{z}(t)) \leq \hat{r}, \qquad \qquad \forall t\in[t_k,t_{k}+N_p\Delta] \notag\\
	&&& \hspace{3cm} \text{if} ~\mathbf{x}(t_k) \in \Omega_{\hat{r}} \label{eq21e}\\
   &&& \dot{V}(\mathbf{z}(t_k),\mathbf{u}(t_k)) \leq \dot{V}(\mathbf{z}(t_k),{h}(\mathbf{z}(t_k))), \notag\\
   &&& \hspace{3cm} \text{if} ~\mathbf{x}(t_k) \in \Omega_r/ \Omega_{\hat{r}}\label{eq21f}
	\end{align}
\end{subequations}  
where $\mathcal{S}(\Delta)$ is the family of piece-wise constant functions with sampling period $\Delta$,  $N_p$ is the prediction horizon, and $W \in \mathbb{R}^{N \times N}$ and $R \in \mathbb{R}$ are positive definite weighting matrices. The manipulated input (solution to the optimization problem) of the above system under the LMPC control law is defined as 
\begin{equation}\label{eq}
{u} = {u}^\star (t|t_k) , \qquad \forall t\in[t_k,t_{k}+N_p\Delta)
\end{equation}
where ${u}^\star(t|t_k) = [{u}^\star(t_k),\cdots,{u}^\star(t_{k}+N_p\Delta)]$. The first value of ${u}^\star(t|t_k)$ is applied to the closed-loop system for the next sampling time period $t \in [t_k,t_k + \Delta)$ and the procedure is repeated until the end of operation. 

In the LMPC formulation of (\ref{eq21a}-\ref{eq21c}), (\ref{eq21a}) denotes a performance index that is to be minimized, (\ref{eq21b}) is the Koopman bilinear model of the system of (\ref{eq11}) used to predict the future evolution of the states, and (\ref{eq21c}) provides the initial condition which is obtained as a transformation of the actual state measurement. In addition to these constraints, the LMPC formulation considers two Lyapunov constraints, \eqref{eq21e} and \eqref{eq21f}. In the design of LMPC, one important factor we need to consider is the sample-and-hold implementation of the control law. To explicitly deal with the sampled system, we consider a region $\Omega_{\hat{r}}$, where $\hat{r} < r$.  Specifically, when $\mathbf{z}(t_k)$ is received at a sampling time $t_k$, if $\mathbf{z}(t_k)$ is within the region $\Omega_{\hat{r}}$, the LMPC minimizes the cost function within the region $\Omega_{\hat{r}}$; however, if $\mathbf{z}(t_k)$ is in the region $\Omega_r/ \Omega_{\hat{r}}$, i.e., $\mathbf{z}(t_k) \in \Omega_r$ but $\mathbf{z}(t_k) \notin \Omega_{\hat{r}}$, the LMPC first drives
the system state to the region $\Omega_{\hat{r}}$ and then minimizes the cost function within $\Omega_{\hat{r}}$. In other words, due to the sample-and-hold implementation of the control law, the region $\Omega_{\hat{r}} \subset \Omega_{r}$ is chosen as a `safe' zone to make $\Omega_{r}$ invariant. Please note that this is not a limitation of the LMPC formulation but of the discrete-time implementation of the control action to a continuous-time dynamical system. Ultimately, the size of the safe set $\Omega_{\hat{r}}$ depends on the hold time (i.e., sampling time), $\Delta$ (details given below in \emph{Proposition \ref{pro2}}).

Therefore, (\ref{eq21e}) is only active when $\mathbf{z}(t_k) \in \Omega_{\hat{r}}$ and ensures that the sampled state is maintained in the region $\Omega_{\hat{r}}$ (so that the actual state of the closed-loop system is in the stability region $\Omega_{r}$).  The constraint (\ref{eq21f}) is only active when $\hat{r} < V(\mathbf{z}(t_k)) \leq r$ and ensures the rate of change of the Lyapunov function is smaller than or equal to that of the value obtained if the explicit control law ${h}(\mathbf{z})$ is applied to the closed-loop system in a sample-and-hold fashion. These constraints allow the LMPC controller to inherit the stability properties of $h(\mathbf{z})$, i.e., it possesses at least the same stability region $\Omega_r$ as the controller ${h}(\mathbf{z})$. This implies that the (equilibrium point of) closed-loop system of \eqref{eq21a}-\eqref{eq21f} is guaranteed to be stable for any initial state inside the region $\Omega_r$ provided that the sampling time $\Delta$ is sufficiently small. Note that because of this property, the LMPC does not require a terminal constraint used in a traditional MPC setting. Additionally, the feasibility of \eqref{eq21a}-\eqref{eq21f} is guaranteed because ${u} = {h}(\mathbf{z})$ is always a feasible solution to the above optimization problem. Even though the above formulation does not explicitly consider the state constraints, they can be readily incorporated. 

\begin{pro}\label{pro2}
	Consider the system of (\ref{eq11}) under the MPC control law of (\ref{eq21a})-(\ref{eq21f}),   which is designed using a CLF, $V$, that has a stability region $\Omega_r$ under continuous implementation of the explicit controller ${h}(\mathbf{z})$. Then, given any  positive real number $d$, $\exists$ positive real numbers $\Delta^\star$ such that if $\mathbf{z}(0) \in \Omega_r$ and $\Delta \in (0,\Delta^\star]$, then $\mathbf{z}(t) \in \Omega_r, \forall t \geq 0$ and $\lim_{t \to \infty} \|\mathbf{z}(t)\| \leq d$.
\end{pro} 
 \begin{proof}
 	The proof is divided into three parts. In \emph{Part \ref{part1}}, the robustness of the explicit controller is shown which preserves the closed-loop stability when the control action is implemented in a sample-and-hold fashion with a sufficiently small hold time ($\Delta$). In \emph{Part \ref{part2}}, the controller of \eqref{eq21a}-\eqref{eq21f} is shown to be feasible for all $\mathbf{z}(0) \in \Omega_r$. Subsequently, in \emph{Part \ref{part3}}, it is shown that the stability region $\Omega_r$ is invariant under the predictive controller of \eqref{eq21a}-\eqref{eq21f}.  
 	\begin{proofpart}\label{part1}
 	 To prove the robustness of the explicit controller, we need to show the existence of a positive real number $\Delta^\star$ such that all state trajectories originating in $\Omega_r$ converge to the level set $\Omega_{\hat{r}}$ for any value of $\Delta \in (0,\Delta^\star]$. To achieve this, we need to consider different cases for $\mathbf{z}(0)$ inside the stability region, i.e., we consider arbitrary regions $\mathcal{Z}$ and $\Omega_{r^\prime}$ inside $\Omega_{r}$.  Figure~\ref{fig1} represents a schematic of the different cases considered in the following proof.
 	 
 	 	\begin{figure}[h]
 	 	\hspace{-3cm}
 	 	\vspace*{-1cm}
 	 	\includegraphics[width=0.8\textwidth]{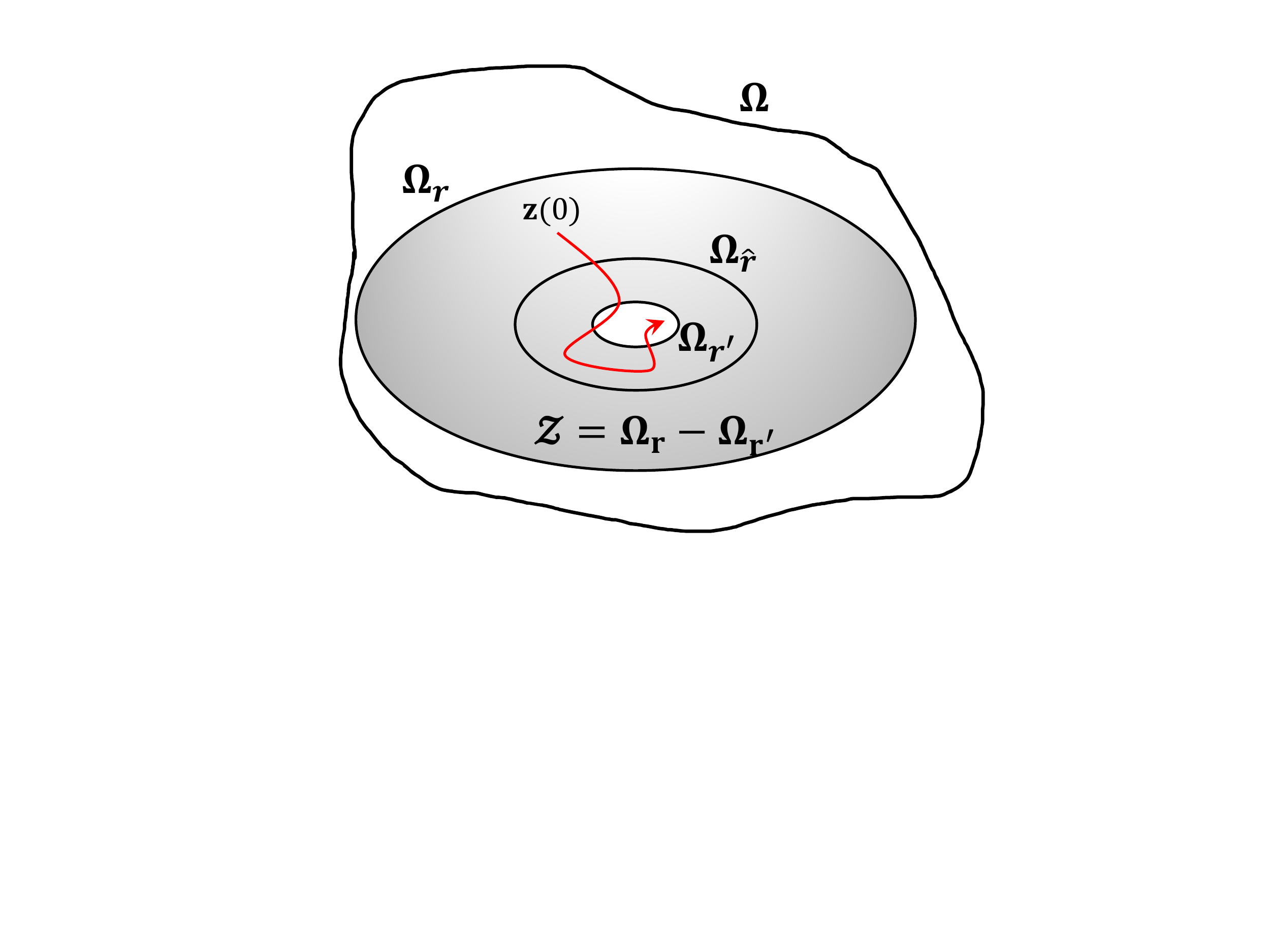}
 	 	\vspace*{-4cm}
 	 	\caption{A schematic representing the stability region of the bounded controller $\Omega_{r}$, together with the sample-and-hold constrained set, $\Omega_{\hat{r}}$, and the overall stability region of the system, $\Omega$. The grey shaded part represents the ring, $\mathcal{Z}$, close to the boundary of the stability region, $\Omega_r$.}\label{fig1}
 	 \end{figure}

 	 First, consider a small region close to the boundary of the stability region denoted as $\mathcal{Z} := \{\mathbf{z}  : (r-r^\prime) \leq V(\mathbf{z}) \leq r\}$, for some $0 < r^\prime < r$. Now, let $h(0) = h_0$ be computed for $\mathbf{z}(0) = \mathbf{z}_0 \in \mathcal{Z}$ and held constant until a time $\hat{\Delta}$ such that $h(t) := h_0 ~\forall t \in (0,\hat{\Delta}]$. Then,
  \begin{equation}\label{eq23}
 \begin{aligned}
 \dot{V}(\mathbf{z}(t)) &= L_\Lambda V(\mathbf{z}(t)) + L_B V(\mathbf{z}(t)) h_0 \\
 &= L_\Lambda V(\mathbf{z}_0) + L_B V(\mathbf{z}_0) h_0 \\
 &~ + (L_\Lambda V(\mathbf{z}(t)) - L_\Lambda V(\mathbf{z}_0)) \\
 &~ + (L_B V(\mathbf{z}(t)) h_0 -  L_B V(\mathbf{z}_0) h_0).
 \end{aligned}
 \end{equation}
Since the initial state $\mathbf{z}_0 \in \mathcal{Z} \subseteq \Omega_r$, and $h_0$ is computed based on the stabilizing control law \eqref{eq19}, it follows that $\dot{V}(\mathbf{z}_0) :=   L_\Lambda V(\mathbf{z}_0) + L_B V(\mathbf{z}_0) h_0 \leq -\rho V(\mathbf{z}_0)$ (this can be shown by substituting \eqref{eq19} in $\dot{V}$). Combining this with the definition of $\mathcal{Z}$, we have $L_\Lambda V(\mathbf{z}_0) + L_B V(\mathbf{z}_0) h_0 \leq -\rho (r-r^\prime)$. 

We also need the following properties to complete the proof. 
\begin{prp}
Since the evolution of $\mathbf{z}$ is continuous, $\|u\| \leq u_{max}$ and $\mathcal{Z}$ is bounded, one can find, for all $\mathbf{z}_0 \in \mathcal{Z}$ and a fixed $\hat{\Delta}$, a positive real number $k_1$ such that $\|\mathbf{z}(t) - \mathbf{z}_0\| \leq k_1\hat{\Delta}$ for all $t \leq \hat{\Delta}$.
\end{prp} 
\begin{prp}
Additionally, since $L_\Lambda V(\cdot)$ and $ L_B V(\cdot)$ are continuous functions, the following properties hold:
\begin{equation}\label{eq24}
\begin{aligned}
&\|L_\Lambda V(\mathbf{z}(t)) - L_\Lambda V(\mathbf{z}_0)\| \leq k_2\|\mathbf{z}(t) - \mathbf{z}_0\| \leq k_1k_2\hat{\Delta} \\
& \|L_B V(\mathbf{z}(t)) h_0 -  L_B V(\mathbf{z}_0) h_0\| \leq k_3\|\mathbf{z}(t) - \mathbf{z}_0\| \leq k_1k_3\hat{\Delta} \\
\end{aligned} 
\end{equation}
\end{prp} 
where the second inequality in each equation holds because of \emph{Property 1}.  Using all the above inequalities in \eqref{eq23}, 
\begin{equation}\label{eq25}
\dot{V}(\mathbf{z}(t)) \leq -\rho(r-r^\prime)  + (k_1k_3 + k_2k_3)\hat{\Delta}
\end{equation}
 Now, if we choose $\hat{\Delta} < (\rho(r-r^\prime) - c)/(k_1k_3 + k_2k_3)$ where $c < \rho(r-r^\prime)$ is a positive number, we get $\dot{V}(\mathbf{z}(t)) \leq -c < 0$ for all $t \leq \hat{\Delta}$. This implies that, given a $\hat{r}$, if we find an $r^\prime$ such that $r-r^\prime < \hat{r}$ and determine the corresponding $\hat{\Delta}$, then the control action computed for any $\mathbf{z} \in \mathcal{Z}$ and held for a time period less than $\hat{\Delta}$ will ensure that the state does not escape $\Omega_r$ (because $\dot{V} < 0$ during this time). 
 
 Now, we need to show the existence of a ${\Delta^\prime}$ such that for all $\mathbf{z}_0 \in \Omega_{r^\prime} := \{\mathbf{z}_0 : V(\mathbf{z}_0) \leq r -r^\prime\} $ we have $\mathbf{z}_0 \in \Omega_{\hat{r}} := \{\mathbf{z}_0 : V(\mathbf{z}_0) \leq \hat{r}\} $. Consider ${\Delta^\prime}$ such that 
 \begin{equation}\label{eq26}
 \hat{r} = \underset{\mathbf{z}_0 \in \Omega_r^\prime, h \in \mathcal{U}, t \in [0,{\Delta^\prime}]}{\text{max}} V(\mathbf{z}(t))
 \end{equation}
  This is possible because both $V$ and $\mathbf{z}$ are continuous functions, and therefore for any $r^\prime < r$, one can find a sufficiently small $\Delta^\prime$ such that \eqref{eq26} holds. All that remains now is to show that for all $\mathbf{z}_0 \in \Omega_{\hat{r}}$ if $\Delta \in (0,\Delta^\star]$ where $\Delta^\star = \text{min}\{\hat{\Delta},\Delta^\prime\}$, then $\mathbf{z}(t) \in \Omega_{\hat{r}} ~\forall t \geq 0$. 
  
  Consider all $\mathbf{z}_0 \in \Omega_{\hat{r}} \cap \Omega_{r^\prime}$. Then by definition, $\mathbf{z}(t) \in \Omega_{\hat{r}}$ for $t \in [0,\Delta^\star]$ since $\Delta^\star \leq {\Delta^\prime}$. On the other hand, for all $\mathbf{z}_0 \in \Omega_{\hat{r}}/\Omega_{r^\prime}$, i.e., $\mathbf{z}_0 \in \mathcal{Z}$, it was shown that $\dot{V} < 0$ for $t \in [0,\Delta^\star]$ since $\Delta^\star \leq \hat{\Delta}$. Therefore, $\Omega_{\hat{r}}$ is an invariant set under the control law of \eqref{eq19}. 
  
  Hence, all trajectories originating in $\Omega_r$ converge to $\Omega_{\hat{r}}$ with a hold time less than $\Delta^\star$. That is, for all $\mathbf{z}_0 \in \Omega_r, \text{lim} ~\text{sup}_{t \to \infty} V(\mathbf{z}(t)) \leq \hat{r}$. Since, $V(\cdot)$ is a continuous function, one can always find a finite, positive number $d$ such that $V(\mathbf{z})\leq\hat{r}\implies\|\mathbf{z}\| \leq d$. Therefore,  $\text{lim}~ \text{sup}_{t \to \infty} V(\mathbf{z}(t)) \leq \hat{r} \implies \text{lim} ~\text{sup}_{t \to \infty} \|\mathbf{z}(t)\| \leq d$.
 	\end{proofpart}
 	\begin{proofpart}\label{part2}
 		Let us consider some $\mathbf{z}(0) \in \Omega_r$ under the predictive controller of \eqref{eq21a}-\eqref{eq21f} with a prediction horizon $N_p$ denoting the number of prediction steps. There are two cases. If $\mathbf{z}_0 \in \Omega_r/\Omega_{\hat{r}}$, the feasibility of constraint \eqref{eq21f} is guaranteed by the control law of \eqref{eq19} as shown in \emph{Part \ref{part1}}. Additionally, if $V(\mathbf{z}(0)) \leq \hat{r}$, once again the control input trajectory under the explicit controller of \eqref{eq19}, given by $u(t) = h(\mathbf{z}(t)), ~\forall t \in [t_k,t_{k}+N_p\Delta]$, provides a feasible initial guess to constraint \eqref{eq21e} because it was designed to stabilize the system, i.e., $V(\mathbf{z}(t)) \leq \hat{r}$. This shows that for all $\mathbf{z}(0) \in \Omega_r$ the Koopman LMPC of \eqref{eq21a}-\eqref{eq21f} is feasible. 
 	\end{proofpart}
 	\begin{proofpart}\label{part3}
 		To prove the last part, please note that since constraint \eqref{eq21f} is feasible, upon implementation it ensures that the value of the Lyapunov function under the predictive controller $u(t)$ decreases at each sampling time. Since $\Omega_r$ is a level set of $V$, and $\dot{V}$ decreases, the state trajectories cannot escape $\Omega_r$. Additionally, satisfying constraint $\eqref{eq21e}$ means that $\Omega_{\hat{r}}$ continues to remain invariant under the implementation of the predictive controller of \eqref{eq21a}-\eqref{eq21f}. The recursive feasibility of \eqref{eq21d}-\eqref{eq21f} implies that $V \leq r$ and $\dot{V} < 0$ for all $\mathbf{z}(t)$ under the Lyapunov-based controller given by \eqref{eq21a}-\eqref{eq21f}. However, since it is implemented in a sample-and-hold fashion there exists a maximum sampling time $\Delta^\star$, given in \emph{Part \ref{part1}}, such that when $\Delta \in (0,\Delta^\star)$ it is guaranteed that for all $\mathbf{z(0)} \in \Omega_r$, $\text{lim}_{t \to \infty} \|\mathbf{z}(t)\| \leq d$. 
 	\end{proofpart}
 	This completes the proof.
 \end{proof}

\begin{rmk}
	Please note that in practice, one can characterize the values of $r,\hat{r},\Delta^\star$ and $d$ by performing several closed-loop simulations where the controller defined in \eqref{eq21a}-\eqref{eq21f} is continuously applied to the system. However, the estimate of the stability region $\Omega_r$ determined using explicit controllers such as \eqref{eq19} does not necessarily equate the entire domain $\Omega$, which remains a difficult problem even for linear systems. Nevertheless, these estimates can be improved by considering multiple CLFs.
\end{rmk}

\emph{Proposition \ref{pro2}} formalizes that the stability properties of the Koopman bilinear system under the Lyapunov–based predictive controller are inherited from the
explicit (bounded) controller under discrete implementation. Now, when there is no mismatch between the Koopman model and the original system, the stability properties will be easily translated to the original system. Obviously, we can derive an exact model without any model-plant mismatch if we can implement the infinite dimensional Koopman operator. However, as described previously, only a finite dimensional approximation based on the projection of these operators on a subspace is commonly used for practical implementation. In this regard, since the model-plant mismatch between the Koopman model and the original system is inevitable, we additionally study and derive the bound on the prediction error between the original state and the predicted state from the Koopman model in the following theorem. 

In order to extend the stability results to the original nonlinear system of \eqref{eq7}, we make the following assumption.

\begin{asm}\label{asm2}
	Let the inverse mapping from the Koopman space, $\mathbf{z}$, to the original state space, $\mathbf{x}$, be continuously differentiable, i.e., $\exists ~\boldsymbol{\xi}(\mathbf{z}) = [\xi_1(\mathbf{z}), \cdots, \xi_n(\mathbf{z})]^T \in C^1: \mathbb{R}^N \to \mathbb{R}^n$ such that $\hat{x}_i = \xi_i(\mathbf{z}), i = 1,\cdots,n$ where $\hat{\mathbf{x}} = [\hat{x}_1,\cdots,\hat{x}_n]$ is the predicted state vector obtained from the inverse mapping defined above.   
\end{asm}

Then, the stability properties of the closed-loop system (\ref{eq21a})-(\ref{eq21f}) of $\mathbf{z}$ can be shown to be inherited to the original nonlinear system of $\mathbf{x}$ under the above assumption and is formalized in the following theorem.

\begin{thm}\label{thm1}
	Suppose that system (\ref{eq7}) satisfies \emph{Assumptions \ref{asm1}-\ref{asm2}}. Let $\mathbf{x}(t)$ and $\hat{\mathbf{x}}(t)$ denote the original state and the predicted state values, respectively. The solutions for $\mathbf{x}(t)$ and $\hat{\mathbf{x}}(t)$ are given by the following dynamic equations:
	\begin{align}
	 & \dot{\mathbf{x}}(t) = \mathbf{f}(\mathbf{x}(t),u(t)), \quad \quad ~~\mathbf{x}(0) = \mathbf{x}_0 \label{eq27}\\
	 & {\hat{\mathbf{x}}}(t) = \xi(\mathbf{z}(t)), \qquad \qquad ~~~\hat{\mathbf{x}}(0) = \mathbf{x}_0 \label{eq28}\\
	 & \dot{\mathbf{z}}(t) = \Lambda\mathbf{z}(t) + u(t)B\mathbf{z}(t), ~\mathbf{z}(0) = \boldsymbol{\phi}(\hat{\mathbf{x}}(0))  \label{eq29}
	\end{align}
	Then, the difference between $\mathbf{x}(t)$ and $\hat{\mathbf{x}}(t)$ is bounded by
	\begin{equation}\label{eq30}
	\|\mathbf{x}(t)-\hat{\mathbf{x}}(t)\|  \leq \frac{\nu}{l_x}(e^{l_xt} - 1)
	\end{equation}
	where $\nu$ denotes the modeling error which bounds the difference between
	\begin{equation}\label{eq31}
     \|\mathbf{f}(\hat{\mathbf{x}},u) - \hat{\mathbf{f}}(\hat{\mathbf{x}},u)\| \leq \nu
	\end{equation}
	where $\mathbf{f}(\cdot) = \mathbf{F}(\cdot) + \mathbf{G}(\cdot)u$ is the original nonlinear dynamical system, and $\hat{\mathbf{f}}(\hat{\mathbf{x}},u) = \frac{\partial \xi}{\partial \mathbf{z}}\dot{\mathbf{z}} $ denotes the solution to $\dot{\hat{\mathbf{x}}}(t)$.
	Under this condition, the stabilizing feedback control input ${u}^\star(t)$ obtained from the Lyapunov-based predictive control law of (\ref{eq21a})-(\ref{eq21f}) for the Koopman linear system of (\ref{eq9}) also stabilizes the original system of (\ref{eq6}), i.e., the origin of the closed-loop system of (\ref{eq6}) is Lyapunov stable. 
\end{thm}

\begin{proof}
	The proof is divided into two parts. First, we show that the predicted state $\hat{\mathbf{x}}(t)$ is stable under the application of the Koopman LMPC controller of \eqref{eq21a}-\eqref{eq21f} to the Koopman bilinear system. In the second part, we show that the evolution of the error between the original state and the predicted state is bounded under \emph{Assumption \ref{asm2}} and the Lipschitz property of the vector fields, $\mathbf{F}$ and $\mathbf{G}$.
	\begin{proofpart}\label{part4}
Let us consider any initial condition $\mathbf{x}(0)$ such that $\mathbf{x}(0) = \hat{\mathbf{x}}(0) = \mathbf{x}_0$ and $\|\mathbf{x}_0\| \leq \delta$. Recall from \emph{Proposition \ref{pro2}} that the predictive controller of \eqref{eq21a}-\eqref{eq21f} ensures that the lifted states do not escape the stability region $\Omega_r$, i.e., $V(\mathbf{z}(t)) \leq r, \dot{V} < 0 ~\forall t$. Therefore, $\text{lim} ~\text{sup}_{t \to \infty} \|\mathbf{z}(t)\| \leq d$. Now, from \emph{Assumption \ref{asm2}}, since the inverse mapping $\boldsymbol{\xi}(\mathbf{z})$ is assumed to be continuous (differentiable), the following holds true:
\begin{equation}
\begin{aligned}
& \|\boldsymbol{\xi}(\mathbf{z}(t))\| := \|\hat{\mathbf{x}}(t)\| \leq \epsilon_z\|\mathbf{z}(t)\| \\
& \text{lim}~\text{sup}_{t \to \infty} \|\hat{\mathbf{x}}(t)\| \leq \hat{d}
\end{aligned}
\end{equation}  
where $\hat{d} = \epsilon_zd$. 
In other words, since the controller ensures asymptotic stability of the lifted state, it implies that $\|\mathbf{z}(t)\|$ is bounded at all times and eventually converges to $d$. This in turn implies that $\hat{\mathbf{x}}(t)$ is bounded at all times, albeit by different constants at different sampling times. Now, if we choose $\hat{\epsilon}$ to be the maximum of all these bounds, then $\|\hat{\mathbf{x}}(t)\| < \hat{\epsilon},~ \forall t$. Hence, for any initial condition $\|\mathbf{x}_0\|  \leq \delta$, the implementation of the predictive controller of \eqref{eq21a}-\eqref{eq21f} guarantees that $\|\hat{\mathbf{x}}(t)\| \leq \hat{\epsilon}, \forall t$. This implies that the predicted states of the original system starting close enough to the equilibrium (at a distance $\delta$) will be maintained close to the equilibrium at all times.  
	\end{proofpart}
\begin{proofpart}\label{part5}
	Now, it remains to prove that the modeling error between the original state vector and the predicted states is bounded at all times for all $\|\mathbf{x}_0\| \leq \delta$. Let us consider the modeling error $e(t) = \mathbf{x}(t) - \hat{\mathbf{x}}(t)$, then the evolution of the error is given as 
	\begin{equation}\label{eq32}
	\begin{aligned}
     \|\dot{e}(t)\| &= \|\dot{\mathbf{x}}(t) - \dot{\hat{\mathbf{x}}}(t)\| \\
     &= \|\mathbf{f}(\mathbf{x},u) - \hat{\mathbf{f}}(\hat{\mathbf{x}},u)\|
	\end{aligned}
	\end{equation}
	where $\mathbf{f}(\mathbf{x},u) = \mathbf{F}(\mathbf{x}) + \mathbf{G}(\mathbf{x})u$ is the nonlinear dynamical system, and $\hat{\mathbf{f}}(\hat{\mathbf{x}},u)$ denotes the evolution of the predicted state $\hat{\mathbf{x}}$, which can be determined from the following Koopman bilinear system:
	\begin{equation}\label{eq33}
	\begin{aligned}
\hat{\mathbf{f}}(\hat{\mathbf{x}},u) &= \frac{\partial \boldsymbol{\xi}}{\partial \mathbf{z}}\dot{\mathbf{z}} \\[5pt]
	\end{aligned}
	\end{equation}
	By adding and subtracting $\mathbf{f}(\hat{\mathbf{x}},u)$ to \eqref{eq32}, we get
	\begin{equation}\label{eq34}
	\begin{aligned}
	\|\dot{e}(t)\| &=  \|\mathbf{f}(\mathbf{x},u) - \mathbf{f}(\hat{\mathbf{x}},u) + \mathbf{f}(\hat{\mathbf{x}},u)- \hat{\mathbf{f}}(\hat{\mathbf{x}},u)\| \\
	&\leq \|\mathbf{f}(\mathbf{x},u) - \mathbf{f}(\hat{\mathbf{x}},u)\| + \|\mathbf{f}(\hat{\mathbf{x}},u)- \hat{\mathbf{f}}(\hat{\mathbf{x}},u)\| 
	\end{aligned}
	\end{equation}
	The Lipschitz property of $\mathbf{f}(\cdot)$, combined with the bounds on $u$, implies that there exists a positive constant $l_x$ such that the following inequality holds for all $\mathbf{x},\mathbf{x}^\prime \in \mathcal{X}$ and $u \in \mathcal{U}$:
	\begin{equation}\label{eq35}
	\begin{aligned}
	\|\mathbf{f}(\mathbf{x},u) - \mathbf{f}(\mathbf{x}^\prime,u)\| \leq l_x \|\mathbf{x} - \mathbf{x}^\prime\|
	\end{aligned}
	\end{equation}
	Additionally, since $\hat{\mathbf{x}}$ is bounded (see \emph{Part \ref{part4}} in the proof of \emph{Theorem \ref{thm1}}), $\mathbf{f}$ is Lipschitz, and the mapping $\boldsymbol{\xi}$ is continuously differentiable, there exists a positive constant $\nu$ such that the second term on the right hand side of the inequality in \eqref{eq34} is bounded by $\nu$. Combining it with \eqref{eq35} we have
	\begin{equation}\label{eq36}
	\begin{aligned}
\|\dot{e}(t)\| &\leq l_x\|\mathbf{x}-\hat{\mathbf{x}}\| + \nu \\
& \leq l_x\|e(t)\| + \nu
	\end{aligned}		
	\end{equation}
	Therefore, given the zero initial condition (i.e., $e(0) = 0$), the upper bound for the norm of the error vector can be determined by integrating \eqref{eq36} as
	\begin{equation}\label{eq37}
	\begin{aligned}
	\int_0^t \frac{\|\dot{e}(\tau)\|}{l_x\|e(\tau)\| + \nu} \leq t
	\end{aligned}
	\end{equation}
	and solving for $\|e(t)\|$ 
	\begin{equation}\label{eq38}
	\|e(t)\| = \|\mathbf{x}(t)-\hat{\mathbf{x}}(t)\|  \leq \frac{\nu}{l_x}(e^{l_xt} - 1)
	\end{equation}
	Finally, since the error between the original and predicted vectors is bounded and that the Koopman LMPC controller of \eqref{eq21a}-\eqref{eq21f} stabilizes the predicted state vector $\|\hat{\mathbf{x}}(t)\| \leq \hat{\epsilon}$, there exists a positive constant $\epsilon$ such that $\|\mathbf{x}(t)\| \leq \epsilon$ for all $t$.
	
	Therefore, for all $\|\mathbf{x}_0\| \leq \delta$ the implementation of the predictive controller of \eqref{eq21a}-\eqref{eq21f} ensures that $\|\mathbf{x}(t)\| \leq \epsilon$ for all $t$, thereby rendering the original nonlinear system stable.
\end{proofpart}    
	This completes the proof.
\end{proof}

\begin{rmk}
	Please note that one cannot guarantee asymptotic stability of the original nonlinear system under the proposed controller because there is always loss of information when transforming the system to a different space.
\end{rmk}

\begin{rmk}
	\emph{Assumption \ref{asm2}} seems restrictive in selecting the types of basis functions to determine the Koopman bilinear models as the inverse of the eigenfunctions is required to be $C^1$. However, in practice, one can numerically obtain a separate mapping from the Koopman space to the original space without actually inverting the eigenfunctions. One example would be to assume the system states $\mathbf{x}$ be contained in the span of $\bar{\mathcal{G}}$, the finite subset of the observable space. This implies that there exists a constant matrix $ {C} \in \mathbb{R}^{n \times N}$ such that $\mathbf{x} = {C}\mathbf{z}$. Then, a convex optimization problem can be solved to determine the relation $C$ \cite{Korda2018a}. In this case, the error of the optimization problem must be certified to be bounded to ensure that the proposed controller successfully stabilizes the closed-loop system.
\end{rmk}

\begin{rmk}
	Please note that in this work we do not consider model-plant mismatch due to uncertainties. In the presence of disturbances, to ensure the robust closed-loop stability of the original system, we have to show the inherent robustness of the KLMPC law of \eqref{eq21} by guaranteeing the robust feasibility and robust positive invariance of the control system (such as in \cite{Allan2017,Rawlings2017}) under a specific prediction error bound. This robust closed-loop stability of KLMPC will be studied as a future work, and the prediction error bound between the original state and the predicted state based on the Koopman model, which derived in \emph{Theorem \ref{thm1}}, would be a great starting point.
\end{rmk}

\section{NUMERICAL EXPERIMENTS}\label{Results}
We applied our results on two illustrative examples: Van der Pol oscillator and a simple pendulum system, showing the performance of our provably-stable
Lyapunov-based predictive controller designed in the Koopman function space. Each example produced closed-loop results that are stable with respect to the original state-space.

\subsection{Van der Pol oscillator}
In our first example, we consider the Van der Pol oscillator which is described by the following equations:
\begin{equation}\label{Van}
\begin{aligned}
& \dot{x}_1 = x_2 \\
& \dot{x}_2 = (1-x_1^2)x_2 - x_1 + u
\end{aligned}
\end{equation}

At $u=0$, the unforced dynamics of the Van der Pol oscillator are characterized by a limit cycle with an unstable equilibrium point at the origin. We will see whether the proposed Koopman LMPC is able to stabilize the system at the origin. First, the data required to build the Koopman bilinear model is generated by simulating the unforced system of \eqref{Van}. The simulations were initialized uniformly over a circle around the origin, and a number of trajectories for $10$ s were collected with a sampling time of $\Delta = 0.01$ s, i.e., $10^3$ time-series samples per trajectory. In the next step, the states were lifted to the high-dimensional space by using monomials of degree $5$ as the dictionary functions $\boldsymbol{\phi}(\mathbf{x}(t))$, i.e., $\boldsymbol{\phi}(\mathbf{x}(t)) = [1,x_1,x_2,x_1^2,x_1x_2,\cdots,x_2^5]^T$. This results in a lifted system of dimension $\mathbf{z} \in \mathbb{R}^{21}$, and the system matrix $\Lambda$ was constructed using the algorithm described in Section \ref{EDMD}. To determine the $B$ matrix in the controlled setting, the relation between the Koopman eigenfunctions and dictionary functions was used as shown in \eqref{eq18}. The derivatives of the eigenfunctions were computed using the symbolic toolbox in MATLAB. This completes the identification of the Koopman bilinear model of \eqref{eq11}. 

Next, the Koopman LMPC developed in Section \ref{KLMPC} was applied to control the system of \eqref{Van} with $N = 21$ eigenfunctions as the new states, $\mathbf{z}$, in the transformed space. The initial condition was chosen randomly around the unstable equilibrium and the control objective was to stabilize the system at the origin. The CLF used to define the explicit stable controller $h(\mathbf{z})$ was obtained by solving the following optimization problem as defined in \cite{Huang2018}:

\begin{equation}\label{eq40}
	\begin{aligned}
	& \underset{\sigma > 0, P = P^T}{\text{min}} 
	&& \sigma - \gamma \text{trace}(PB)\\[7pt]
	&~~~~\text{s.t}
	&& \sigma I - (PA + A^TP) \geq 0\\
	&&& {c}^LI \leq P \leq {c}^UI
	\end{aligned}
\end{equation} 
where $\sigma$ represents the epigraph form of the largest singular value of $(PA + A^TP)$, and ${c^L, c^U} > 0$ are two positive scalars used to bound the eigenvalues of $P$. The weighting parameter $\gamma > 0$ was chosen as $2$ in this example. 
The explicit controller, $h(\mathbf{z})$, was determined by using the obtained CLF, $V = \mathbf{z}^TP\mathbf{z}$, within the Sontag's formula as shown in \eqref{eq19}. The matrices $W$ and $R$ in \eqref{eq21a} were chosen to be $W = I \in \mathbb{R}^{21 \times 21}$ and $R = 1$, respectively. The prediction horizon was set to $1~s$, i.e., $N_p = 1/\Delta = 100$. Figure~\ref{fig2} shows the comparison between open and closed loop results. It can be observed from Figure~\ref{fig2} that the system was stabilized at the origin as desired.

	\begin{figure}[t]
		\centering     
		\subfigure[]{\label{fig:fig8a}\includegraphics[width=0.4\textwidth]{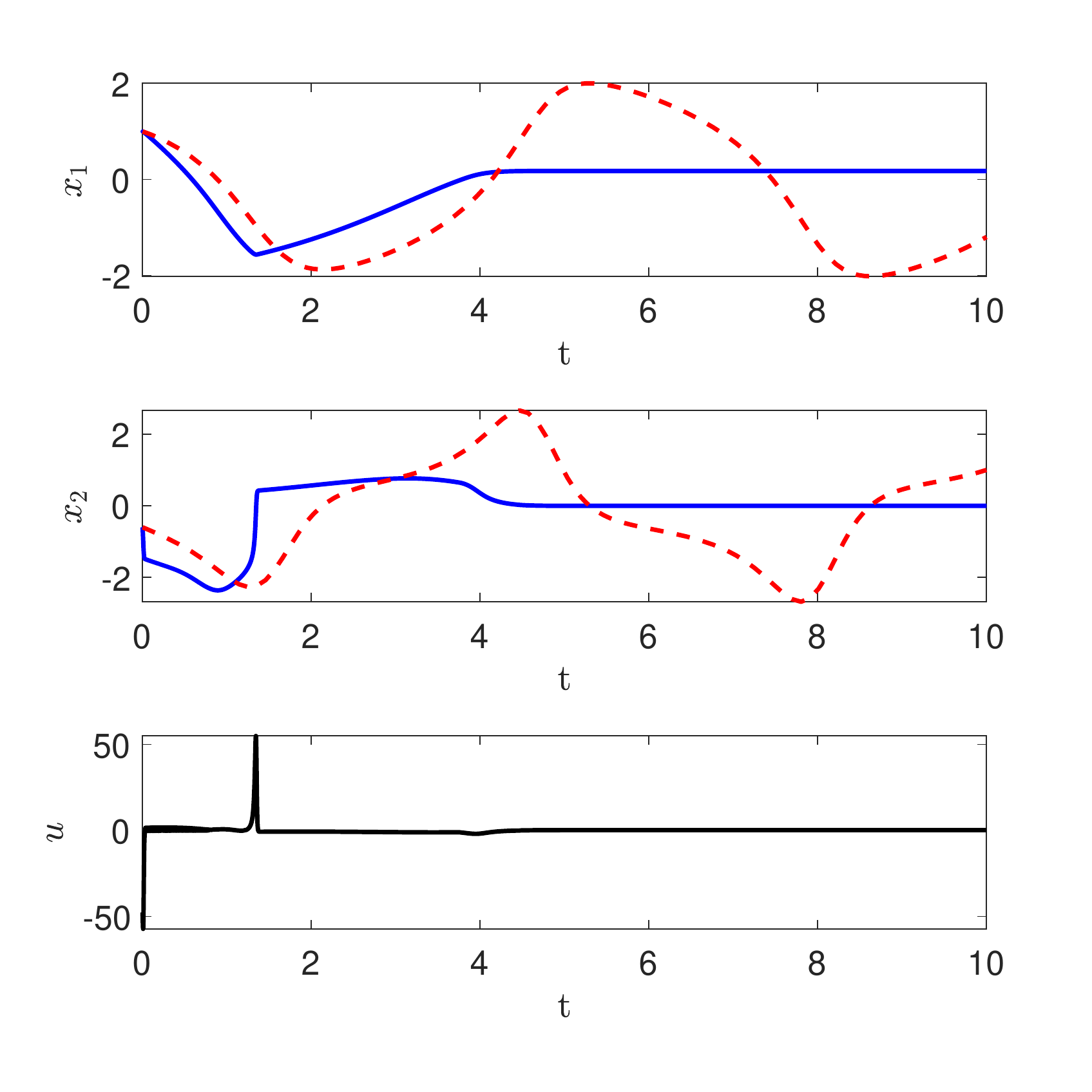}}
		\subfigure[]{\label{fig:fig8b}\includegraphics[width=0.4\textwidth]{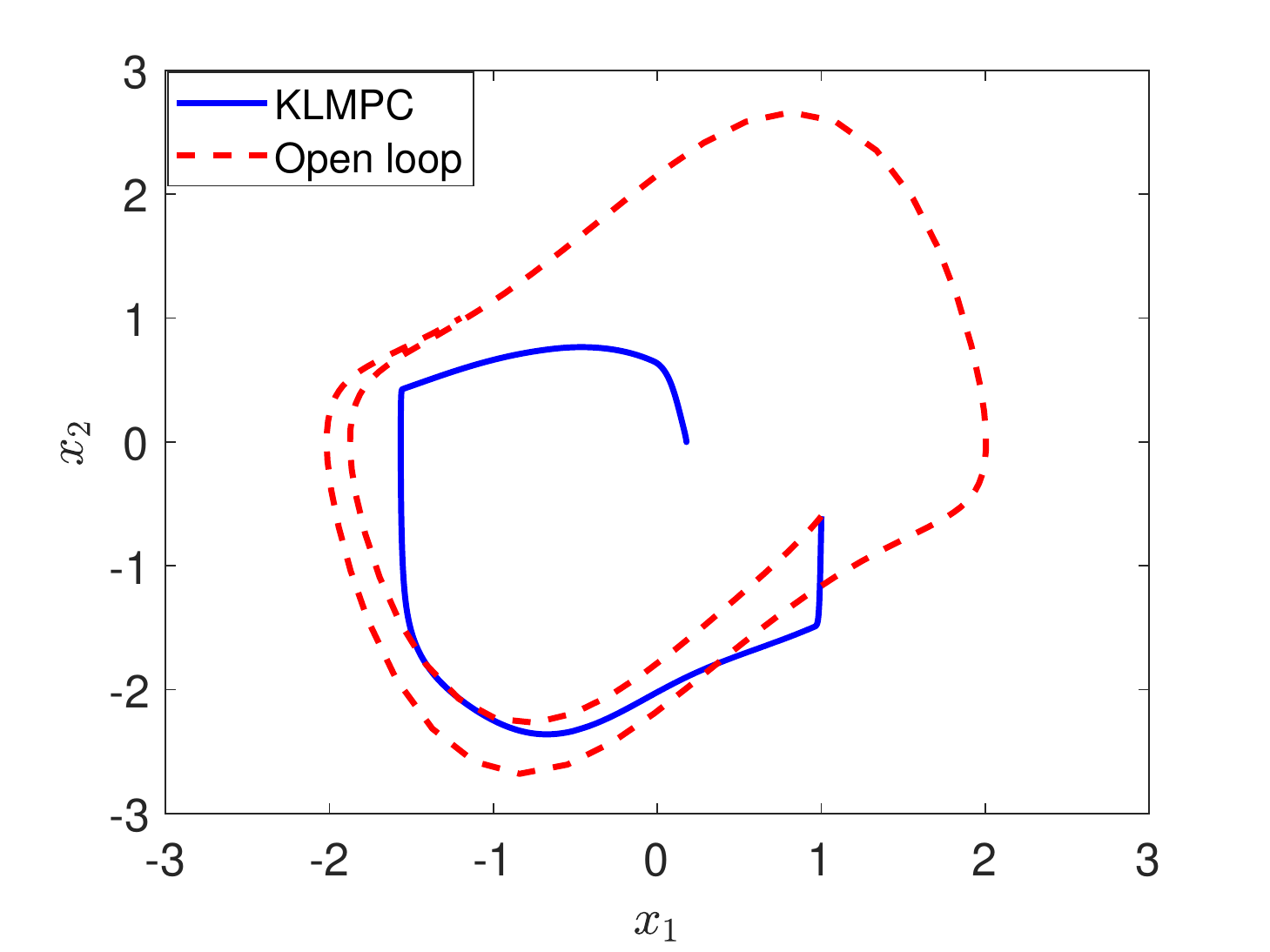}}
		\caption{Comparison of open-loop and closed-loop trajectories for the Van
			der Pol oscillator with $u$ from \eqref{eq21a} - \eqref{eq21f}.}\label{fig2}
		\vspace{-0.5cm}
	\end{figure}

\subsection{Simple pendulum}
The next example we considered is the controlled two dimensional pendulum oscillator given by the following dynamics:

\begin{equation}\label{pen}
\begin{aligned}
& \dot{x}_1 = x_2 \\
& \dot{x}_2 = 0.01x_2 - sin(x_1) + u
\end{aligned}
\end{equation}
where $[x_1,x_2] = [\theta,\dot{\theta}] \in \mathbb{R}^2$ denote the angular displacement and angular velocity of the pendulum, respectively. The system of \eqref{pen} is characterized by a unique unstable equilibrium point at the origin. We considered the system dynamics near the unique unstable equilibrium point at the origin all the way until the limit cycle (shown in Fig.~\ref{fig3}). The Koopman models have to make predictions over this range of initial conditions, and the control objective is to stabilize the system at the origin. 

The training data were generated by simulating the unforced pendulum equation from uniform random initial conditions $(x_1(0),x_2(0)) \in [-2,2] \times [-2,2]$. From each trajectory,  $10^3$ samples were recorded at $\Delta = 0.01$ s apart. Similar to the previous example, the dictionary of observable functions required for nonlinear transformation was considered to be monomials of degree up to $5$, i.e., $\mathbf{z} \in \mathbb{R}^{21}$. The approximation of the Koopman operator and eigenfunctions was then performed by lifting the time-series data samples using the selected dictionary. The system matrices, $\Lambda$, and the control matrix, $B$, were then used to design the feedback controller proposed in \eqref{eq21}. The CLF used in the explicit control design was determined by solving the optimization problem of \eqref{eq40} using the \texttt{cvx} package, a MATLAB-based modeling system for solving disciplined convex optimization problems and is much suitable for semidefinite matrix optimization problems like \eqref{eq40}.   

It is worth mentioning that the proposed Koopman LMPC controller design is not restricted to using a specific form of control law for $h(\mathbf{z})$. In fact, besides Sontag's formula, there are several other possible choices for the explicit controller $h(\mathbf{z})$. Provided we are not constrained to specifications on the amplitude of feedback, we can use the following simple feedback law to define the control law: $h(\mathbf{z}) = -kL_BV(\mathbf{z}) = -k\mathbf{z}^T(PB + B^TP)\mathbf{z}$ . In this example, the value of $k$ was chosen to be $k=10$. The matrices $W$ and $R$ in \eqref{eq21a} were chosen to be $W = I \in \mathbb{R}^{21 \times 21}$ and $R = 1$, respectively. The prediction horizon was set to $1~s$, i.e., $N_p = 1/\Delta = 100$. For the closed-loop simulation, we randomly selected initial points within $[-1, 1]\times[-1, 1]$ and solved the closed-loop system with \texttt{ode45} solver in MATLAB. Figure~\ref{fig3} shows the comparison between open and closed loop results for one such initial condition. It can be observed from Figure~\ref{fig3}  that the controller forced the trajectory of the closed-loop system to the origin as desired. Moreover, in the case of pendulum system, the limit cycle of the open loop system corresponds to the boundary of the basin of attraction and the proposed Koopman LMPC controller forced the states to remain inside this stability region (limit cycle) at all times before the trajectories slide to the origin.
	\begin{figure}[t]
	\centering     
	\subfigure[]{\label{fig:fig3a}\includegraphics[width=0.4\textwidth]{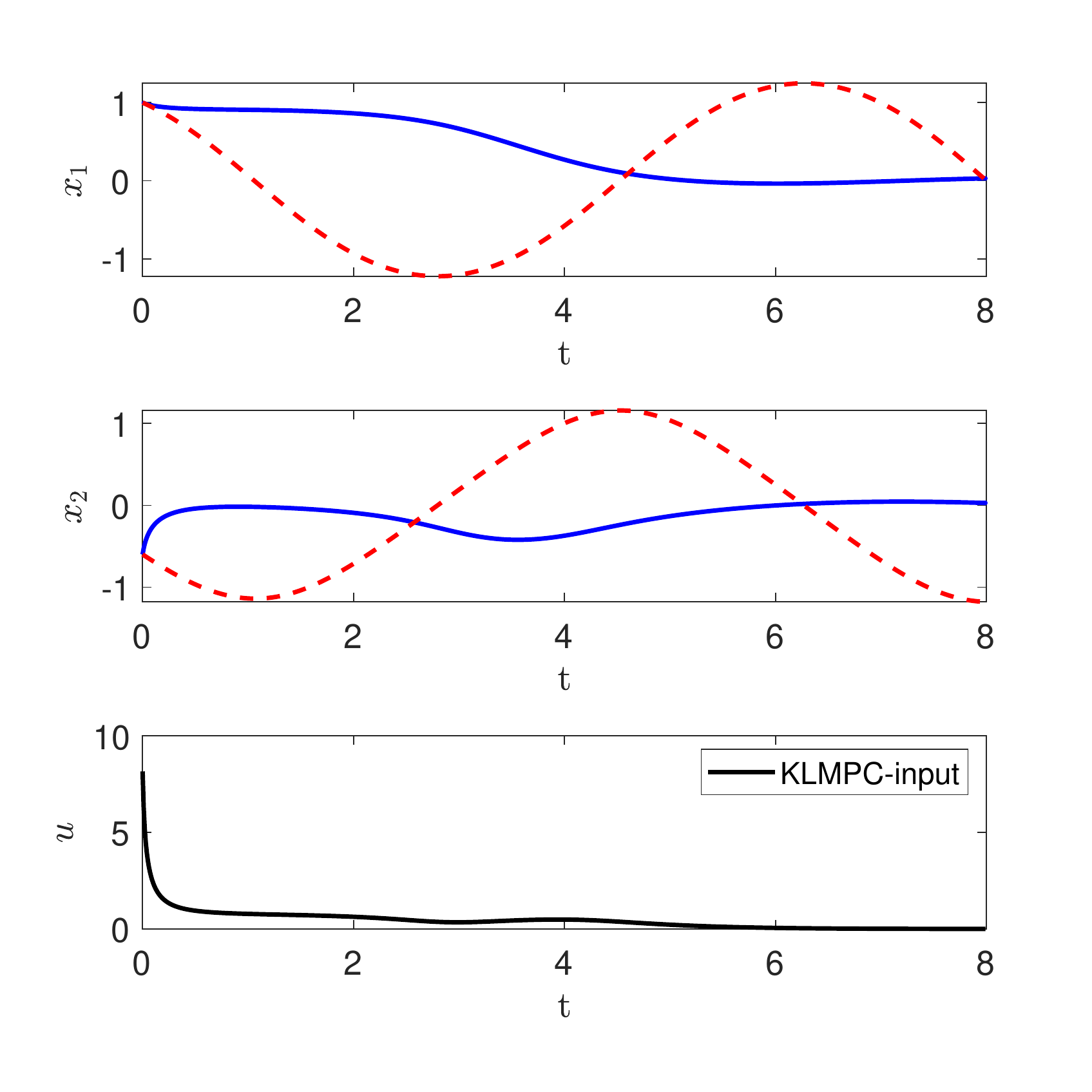}}
	\subfigure[]{\label{fig:fig3b}\includegraphics[width=0.4\textwidth]{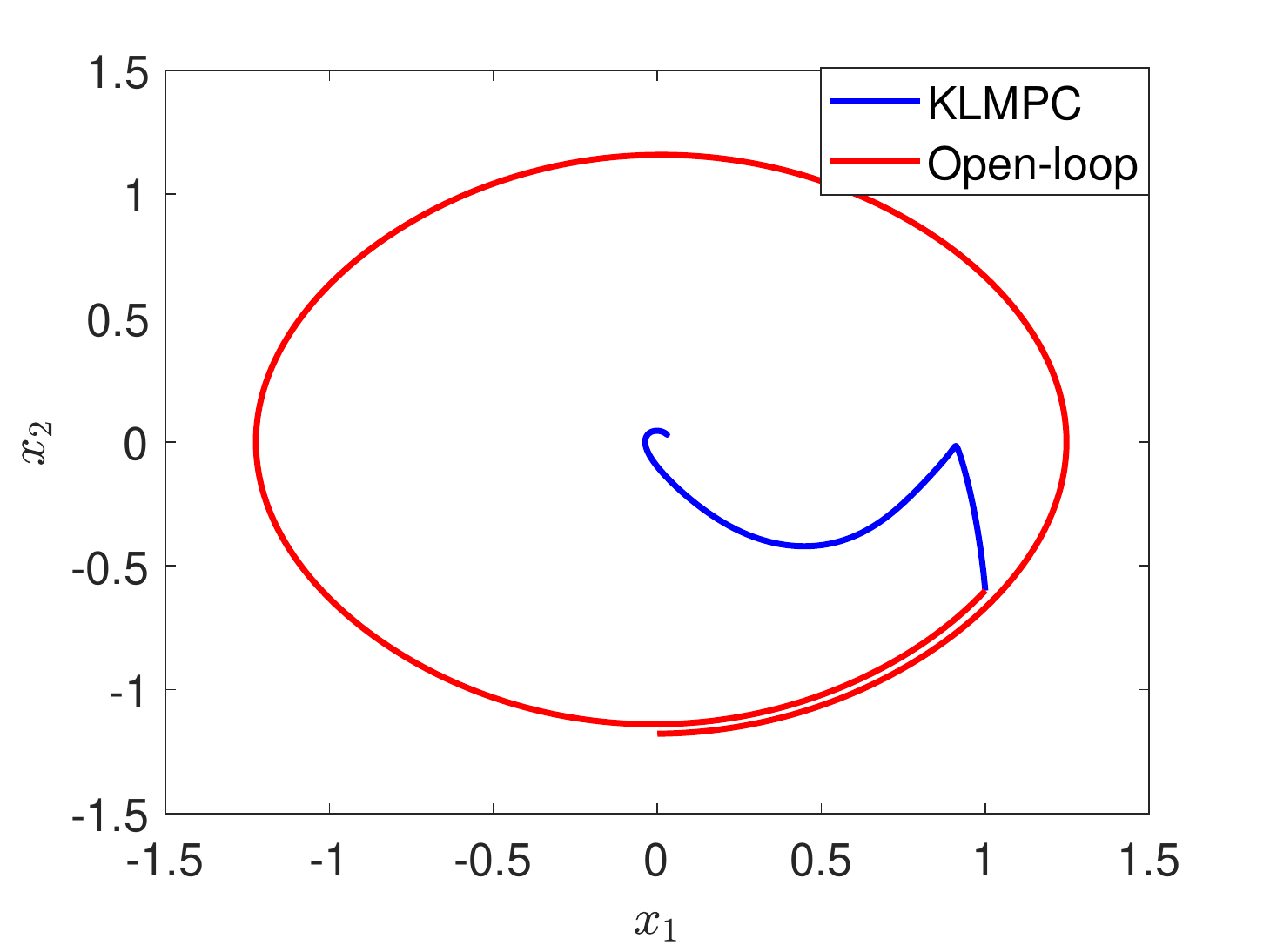}}
	\caption{Comparison of open-loop and closed-loop trajectories for the simple pendulum oscillator with $u$ from \eqref{eq21a} - \eqref{eq21f}.}\label{fig3}
	\vspace{-0.5cm}
\end{figure}
    
\section{CONCLUSIONS}\label{Conclusions}
In this manuscript, we introduced a new approach for designing stabilizing feedback controllers for nonlinear dynamical systems. Leveraging Koopman operator theory, nonlinear dynamics are lifted to a function space where they are embedded in bilinear models that are computed using finite-dimensional approximations to the Koopman operator and its eigenfunctions. A feedback controller is then designed using LMPC that uses explicit Lyapunov constraints to characterize closed-loop stability of the Koopman bilinear system. Due to the bilinear structure of the Koopman model, the CLF can be obtained easily by limiting the search to the class of quadratic functions via an optimization problem. Furthermore, universal control approaches like Sontag's formula readily provides the explicit control law required in the LMPC formulation which is typically a bottleneck for general nonlinear systems. Most importantly, we demonstrated, based on the stability of the Koopman model, that the proposed controller was capable of stably regulating nonlinear dynamics in the original state-space provided that a continuously differentiable inverse mapping exists. The numerical examples indicated that the proposed feedback controller was able to successfully force unstable dynamics to the origin. This was observed from the closed-loop plots presented. Future work will focus on certifying the proposed approach in terms of robustness in the presence of uncertainties. Furthermore, we hope to apply the proposed approach to other flow control problems, studying whether it can provide similar insight into how to design stabilizing feedback controllers for other applications.

\bibliographystyle{IEEEtran}

\begin{IEEEbiography}[{\includegraphics[width=1in,height=1.25in]{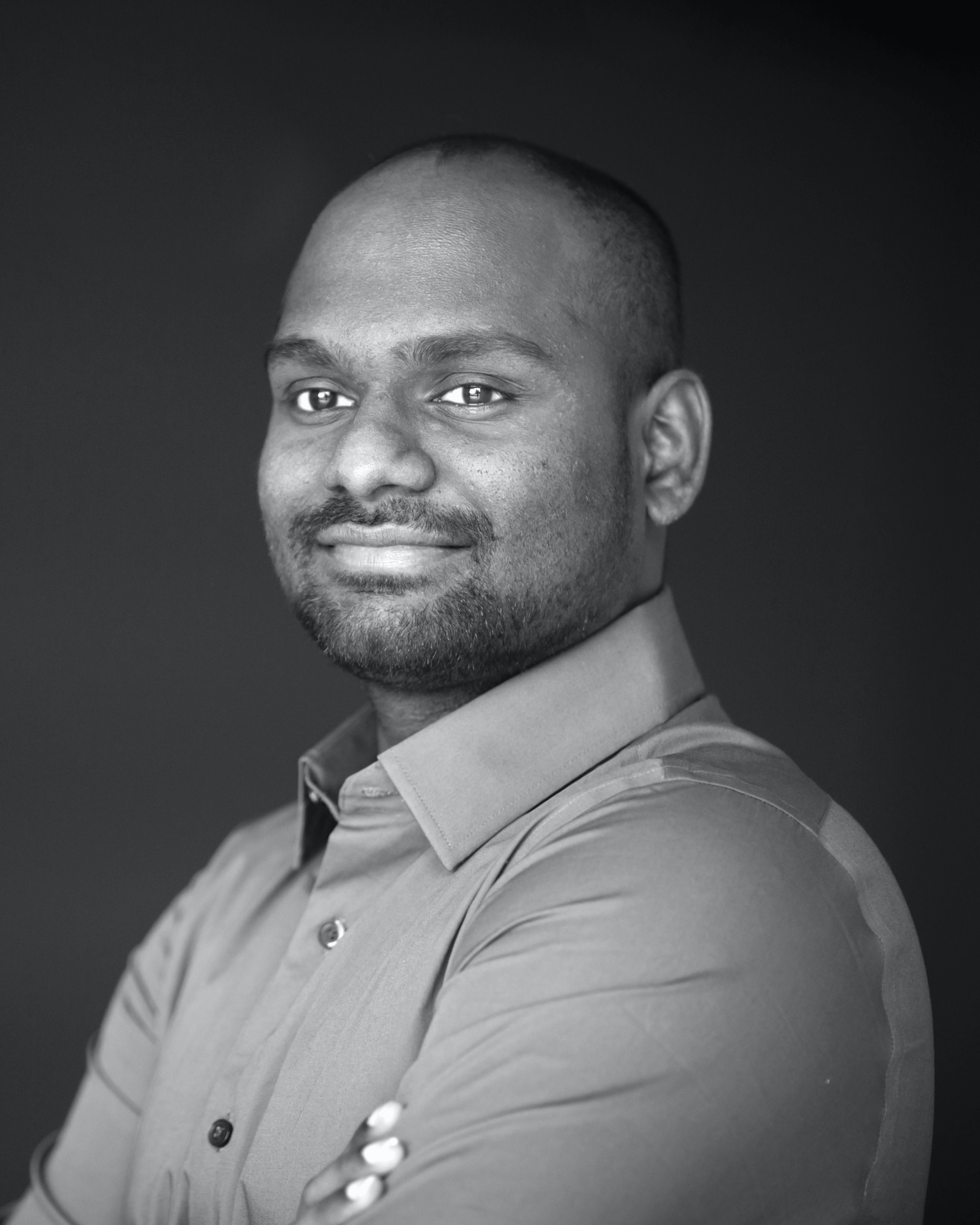}}]
	{Abhinav Narasingam}
	was born in Hyderabad, India in 1991. He received the B.Tech degree in Chemical Engineering from the Indian Institute of Technology Madras in 2013. He is currently pursuing his Ph.D. degree in Chemical Engineering from the Texas A\&M University (TAMU), College Station. 
	
	He is an active member of American Institute of Chemical Engineers (AIChE). He was selected for the Computing \& Systems Technology (CAST) division's directors' student presentation award finals at AIChE Annual Meeting 2019. His research focuses on data-driven reduced-order modeling, applications of operator theoretic methods to dynamical systems and model predictive control of distributed parameter systems. 
\end{IEEEbiography}
\vspace{-1cm}
\begin{IEEEbiography}[{\includegraphics[width=1in,height=1.25in]{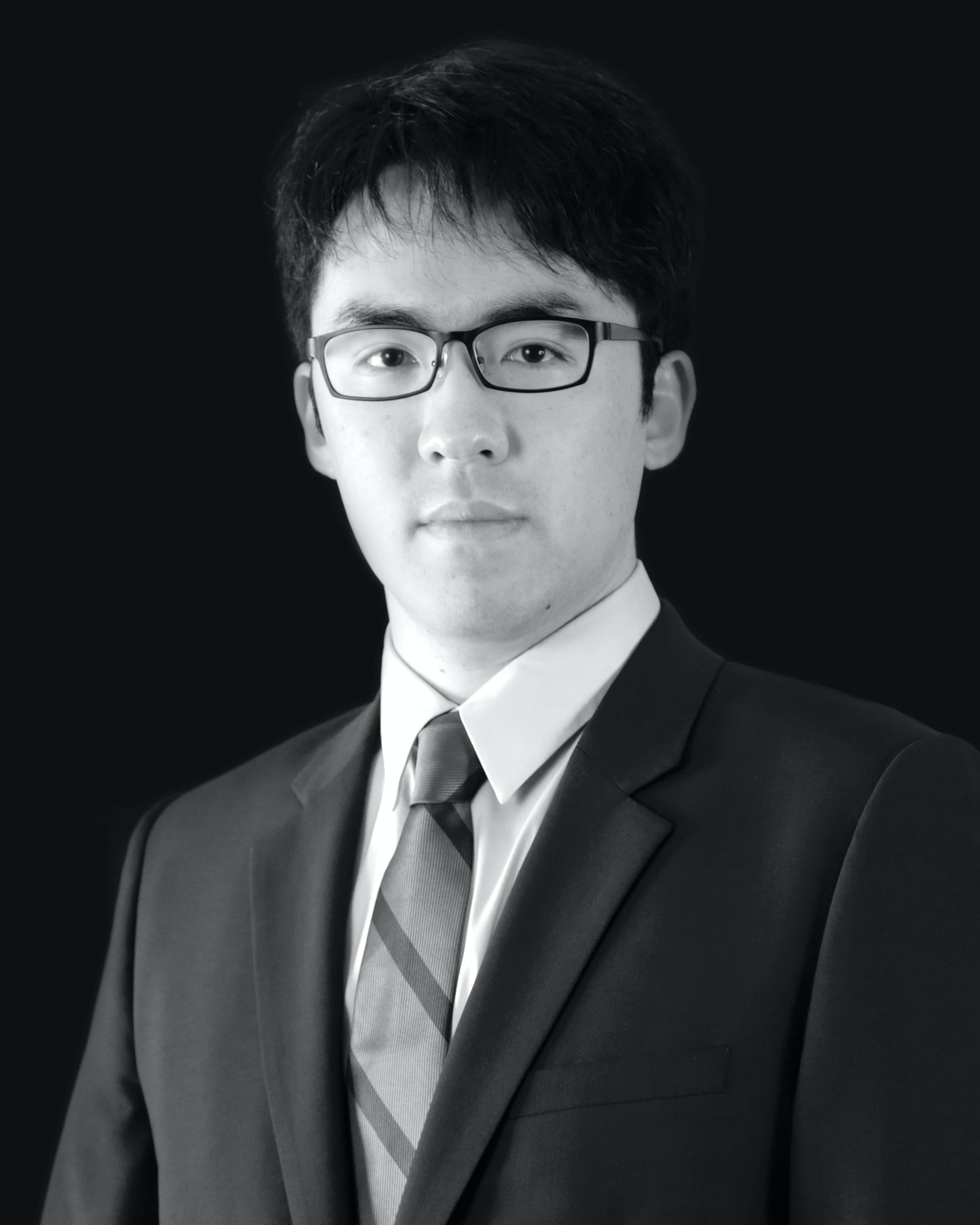}}]
	{Joseph S. Kwon}
 was born in Secaucus, NJ, USA, in 1987. He received the B.S. degrees in Chemical Engineering, Mathematics, and Chemistry (minor) from the University of Minnesota, Twin Cities, in 2009. Then, he received the M.S. degree in Electrical Engineering from the University of Pennsylvania, in 2011, and the Ph.D. degree in Chemical Engineering from the University of California, Los Angeles, in 2015. 

He joined the Department of Chemical Engineering at Texas A\&M University (TAMU), where he is currently an Assistant Professor. His research focuses on multiscale modeling, computation and control of chemical and biological processes with a specialization in oil and gas processes. He is the author of more than 60 peer-reviewed journal publications. 

He has received several awards for his teaching and research work including a President Young Investigator Award from the Korean Institute of Chemical Engineers (KIChE) in 2019, the 2020 Distinguish Teaching Award by the Department of Chemical Engineering at TAMU, the 2020 TEES Young Faculty Fellow Award by the College of Engineering, and a Young Investigator Grant by the Korean-American Scientists and Engineers Association (KSEA) in 2020. He is an Associate Editor of Journal Control, Automation and Systems, and Frontiers in Chemical Engineering. Recently, he has been elected as 10E (Data and Information Systems) programming coordinator for the Computing \& Systems Technology (CAST) Division in the American Institute of Chemical Engineers (AIChE) for 2022. 

\end{IEEEbiography}

\end{document}